\documentclass[a4paper,UKenglish]{oasics-v2016}

\usepackage{microtype}
\usepackage{dsfont} 

\newcommand{\LSH}{\mathcal{H}}
\newcommand{\E}{\mathrm{E}}
\newcommand{\Var}{\mathrm{Var}}

\DeclareMathOperator*{\argmin}{arg\,min}

\newcommand{\1}{\mathds{1}}

\newcommand{\norm}[1]{\left\lVert #1 \right\rVert}

\DeclareMathOperator{\dist}{dist}

\title{Fast Locality-Sensitive Hashing Frameworks for Approximate Near Neighbor Search\footnote{The research leading to these results has received funding from the European Research Council under the European Union’s 7th Framework Programme (FP7/2007-2013) / ERC grant agreement no.~614331.}}

\author[1]{Tobias Christiani}
\affil[1]{IT University of Copenhagen\\
  \texttt{tobc@itu.dk}}
\authorrunning{T. Christiani}

\Copyright{Tobias Christiani}

\begin{document}

\maketitle

\begin{abstract}
The Indyk-Motwani Locality-Sensitive Hashing (LSH) framework (STOC 1998) is a general technique for constructing a data structure to answer approximate near neighbor queries by using a distribution $\LSH$ over locality-sensitive hash functions that partition space.
For a collection of $n$ points, after preprocessing, the query time is dominated by $O(n^{\rho} \log n)$ evaluations of hash functions from $\LSH$ and $O(n^{\rho})$ hash table lookups and distance computations where $\rho \in (0,1)$ is determined by the locality-sensitivity properties of $\LSH$. 
It follows from a recent result by Dahlgaard et al.~(FOCS 2017) that the number of locality-sensitive hash functions can be reduced to $O(\log^2 n)$, leaving the query time to be dominated by $O(n^{\rho})$ distance computations and $O(n^{\rho} \log n)$ additional word-RAM operations.
We state this result as a general framework and provide a simpler analysis showing that the number of lookups and distance computations closely match the Indyk-Motwani framework, making it a viable replacement in practice. 
Using ideas from another locality-sensitive hashing framework by Andoni and Indyk (SODA 2006) we are able to reduce the number of additional word-RAM operations to $O(n^\rho)$.
\end{abstract}
\section{Introduction}
The approximate near neighbor problem is the problem of preprocessing a collection $P$ of $n$ points in a space $(X, \dist)$ into a data structure such that, for parameters $r_1 < r_2$ and given a query point $q \in X$, if there exists a point $x \in P$ with $\dist(q, x) \leq r_1$, then the data structure is guaranteed to return a point $x' \in P$ such that $\dist(q, x') < r_2$.

Indyk and Motwani \cite{indyk1998} introduced a general framework for constructing solutions to the approximate near neighbor problem using a technique known as locality-sensitive hashing~(LSH). 
The framework takes a distribution over hash functions $\LSH$ with the property that near points are more likely to collide under a random $h \sim \LSH$. 
During preprocessing a number of locality-sensitive hash functions are sampled from $\LSH$ and used to hash the points of $P$ into buckets. 
The query algorithm evaluates the same hash functions on the query point and looks into the associated buckets to find an approximate near neighbor.

The locality-sensitive hashing framework of Indyk and Motwani has had a large impact in both theory and practice (see surveys \cite{andoni2008} and \cite{wang2014} for an introduction), and many of the best known solutions to the approximate near neighbor problem in high-dimensional spaces, such as Euclidean space \cite{andoni2006}, the unit sphere under inner product similarity \cite{andoni2015practical}, and sets under Jaccard similarity \cite{broder2000} come in the form of families of locality-sensitive hash functions that can be plugged into the Indyk-Motwani LSH framework. 
\begin{definition}[Locality-sensitive hashing {\cite{indyk1998}}]\label{def:lsh}
Let $(X, \dist)$ be a distance space and let $\LSH$ be a distribution over functions $h \colon X \to R$.
We say that $\LSH$ is $(r_1, r_2, p_1, p_2)$-sensitive if for $x, y \in X$ and $h \sim \LSH$ we have that:
\begin{itemize}
\item If $\dist(x, y) \leq r_1$ then $\Pr[h(x) = h(y)] \geq p_1$.
\item If $\dist(x, y) \geq r_2$ then $\Pr[h(x) = h(y)] \leq p_2$.
\end{itemize}
\end{definition}

\noindent
The Indyk-Motwani framework takes a $(r_1, r_2, p_1, p_2)$-sensitive family $\LSH$ and constructs a data structure that solves the approximate near neighbor problem for parameters $r_1 < r_2$ with some positive constant probability of success. 
We will refer to this randomized approximate version of the near neighbor problem as the $(r_1, r_2)$-near neighbor problem, where we require queries to succeed with probability at least $1/2$ (see Definition \ref{def:ann}).  
To simplify the exposition we will assume throughout the introduction, unless otherwise stated, that $0 < p_1 < p_2 < 1$ are constant, that a hash function $h \in \LSH$ can be stored in $n / \log n$ words of space, and for $\rho = \log(1/p_1) / \log(1/p_2) \in (0,1)$ that a point $x \in X$ can be stored in $O(n^\rho)$ words of space.
The assumption of a constant gap between $p_1$ and $p_2$ allows us to avoid performing distance computations by instead using the $1$-bit sketching scheme of Li and K{\"o}nig~\cite{li2011theory} together with the family $\LSH$ to approximate distances (see Section \ref{sec:sketching} for details).
In the remaining part of the paper we will state our results without any such assumptions to ensure, for example, that our results hold in the important case where $p_1, p_2$ may depend on $n$ or the dimensionality of the space~\cite{andoni2006, andoni2015practical}.  
\begin{theorem}[Indyk-Motwani {\cite{indyk1998, har-peled2012}, simplified}]\label{thm:lsh_im_simple}
Let $\LSH$ be $(r_1, r_2, p_1, p_2)$-sensitive and let $\rho = \frac{\log(1/p_1)}{\log(1/p_2)}$, then there exists a solution to the $(r_1, r_2)$-near neighbor problem using $O(n^{1+\rho})$ words of space and with query time dominated by $O(n^{\rho} \log n)$ evaluations of functions from~$\LSH$.
\end{theorem}
The query time of the Indyk-Motwani framework is dominated by the number of evaluations of locality-sensitive hash functions. 
To make matters worse, almost all of the best known and most widely used locality-sensitive families have an evalution time that is at least linear in the dimensionality of the underlying space~\cite{broder2000, charikar2002, datar2004, andoni2006, andoni2015practical}.
Significant effort has been devoted to the problem of reducing the evaluation complexity of locality-sensitive hash families~\cite{terasawa2007spherical, eshgi2008, dasgupta2011fast, andoni2015practical, kennedy2017fast, shrivastava2016simple, shrivastava2017optimal, dahlgaard2017fast}, while the question of how many independent locality-sensitive hash functions are actually needed to solve the $(r_1, r_2)$-near neighbor problem has received relatively little attention~\cite{andoni2006efficient, dahlgaard2017fast}.  

This paper aims to bring attention to, strengthen, generalize, and simplify results that reduce the number of locality-sensitive hash functions used to solve the $(r_1, r_2)$-near neighbor problem. In particular, we will extract a general framework from a technique introduced by Dahlgaard et al.~\cite{dahlgaard2017fast} in the context of set similarity search under Jaccard similarity, showing that the number of locality-sensitive hash functions can be reduced to $O(\log^2 n)$ in general. 
We further show how to reduce the word-RAM complexity of the general framework from $O(n^\rho \log n)$ to $O(n^\rho)$ by combining techniques from Dahlgaard et al. and Andoni and Indyk~\cite{andoni2006efficient}.
Reducing the number of locality-sensitive hash functions allows us to spend time $O(n^\rho / \log^2 n)$ per hash function evaluation without increasing the overall complexity of the query algorithm --- something which is particularly useful in Euclidean space where the best known LSH upper bounds offer a tradeoff between the $\rho$-value that can be achieved and the evaluation complexity of the locality-sensitive hash function~\cite{andoni2006, andoni2015practical, kennedy2017fast}.
\subsection{Related work}
\subparagraph*{Indyk-Motwani.}
The Indyk-Motwani framework uses $L = O(n^\rho)$ independent partitions of space, each formed by overlaying $k = O(\log n)$ random partitions induced by $k$ random hash functions from a locality-sensitive family $\LSH$.
The parameter $k$ is chosen such that a random partition has the property that a pair of points $x,y \in X$ with $\dist(x, y) \leq r_1$ has probability $n^{-\rho}$ of ending up in the same part of the partition, while a pair of points with $\dist(x, y) \geq r_2$ has probability $n^{-1}$ of colliding. 
By randomly sampling $L = O(n^\rho)$ such partitions we are able to guarantee that a pair of near points will collide with constant probability in at least one of them. 
Applying these $L$ partitions to our collection of data points $P$ and storing the result of each partition of $P$ in a hash table we obtain a data structure that solves the $(r_1, r_2)$-near neighbor problem as outlined in Theorem \ref{thm:lsh_im_simple} above. 
Section \ref{sec:frameworks} and \ref{sec:im} contains a more complete description of LSH-based frameworks and the Indyk-Motwani framework.

\subparagraph*{Andoni-Indyk.}
As previously mentioned, many locality-sensitive hash functions happen to have a super-constant evaluation time. 
This motivated Andoni and Indyk to introduce a replacement to the Indyk-Motwani framework in a paper on substring near neighbor search~\cite{andoni2006efficient}.  
The key idea is to re-use hash functions from a small collection of size $m \ll L$ by forming all combinations of $\binom{m}{t}$ hash functions. 
This technique is also known as tensoring and has seen some use in the work on alternative solutions to the approximate near neighbor problem, in particular the work on locality-sensitive filtering~\cite{dubiner2010bucketing, becker2016, christiani2017framework}.
By applying the tensoring technique the Andoni-Indyk framework reduces the number of hash functions to $O(\exp(\sqrt{\rho \log n \log \log n})) = n^{o(1)}$ as stated in Theorem \ref{thm:lsh_ai_simple}.
\begin{theorem}[Andoni-Indyk {\cite{andoni2006efficient}}, simplified]\label{thm:lsh_ai_simple}
	Let $\LSH$ be $(r_1, r_2, p_1, p_2)$-sensitive and let $\rho = \frac{\log(1/p_1)}{\log(1/p_2)}$, then there exists a solution to the $(r_1, r_2)$-near neighbor problem using $O(n^{1+\rho})$ words of space and with query time dominated by $O(\exp(\sqrt{\rho \log n \log \log n}))$ evaluations of functions from $\LSH$ and $O(n^\rho)$ other word-RAM operations.
\end{theorem}

The paper by Andoni and Indyk did not state this result explicitly as a theorem in the same form as the Indyk-Motwani framework; the analysis made some implicit restrictive assumptions on $p_1, p_2$ and ignored integer constraints. 
Perhaps for these reasons the result does not appear to have received much attention, although it has seen some limited use in practice~\cite{sundaram2013streaming}.
In Section \ref{sec:ai} we present a slightly different version of the Andoni-Indyk framework together with an analysis that satisfies integer constraints, providing a more accurate assessment of the performance of the framework in the general, unrestricted case.   

\subparagraph*{Dahlgaard-Knudsen-Throup.}
The paper by Dahlgaard et al.~\cite{dahlgaard2017fast} introduced a different technique for constructing the $L$ hash functions/partitions from a smaller collection of $m$ hash functions from $\LSH$. 
Instead of forming all combinations of subsets of size $t$ as the Andoni-Indyk framework they instead sample $k$ hash functions from the collection to form each of the $L$ partitions.
The paper focused on a particular application to set similarity search under Jaccard similarity, and stated the result in terms of a solution to this problem. 
In Section \ref{sec:dkt} we provide a simplified and tighter analysis to yield a general framework:
\begin{theorem}[Dahlgaard-Knudsen-Thorup {\cite{dahlgaard2017fast}}, simplified]\label{thm:lsh_dkt_simple}
	Let $\LSH$ be $(r_1, r_2, p_1, p_2)$-sensitive and let $\rho = \frac{\log(1/p_1)}{\log(1/p_2)}$, then there exists a solution to the $(r_1, r_2)$-near neighbor problem using $O(n^{1+\rho})$ words of space and with query time dominated by $O(\log^2 n)$ evaluations of functions from $\LSH$ and $O(n^\rho \log n)$ other word-RAM operations.
\end{theorem}
The analysis of \cite{dahlgaard2017fast} indicates that the Dahlgaard-Knudsen-Thorup framework, when compared to the Indyk-Motwani framework, would use at least $50$ times as many partitions (and a corresponding increase in the number of hash table lookups and distance computations) to solve the $(r_1, r_2)$-near neighbor problem with success probability at least $1/2$. 
Using elementary tools, the analysis in this paper shows that we only have to use twice as many partitions as the Indyk-Motwani framework to obtain the same guarantee of success.

\subparagraph*{Number of hash functions in practice.}
To provide some idea of what the number of hash functions $H$ used by the different frameworks would be in practice, Figure \ref{fig:comparison} shows the value of $\log_2 H$ that is obtained by actual implementations of the Indyk-Motwani (IM), Andoni-Indyk (AI), and Dahlgaard-Knudsen-Thorup (DKT) frameworks according to the analysis in Section \ref{sec:frameworks} for $p_1 = 1/2$ and every value of $0 < p_2 < 1/2$ for a solution to the $(r_1, r_2)$-near neighbor problem on a collection of $n = 2^{30}$ points.
\begin{figure} 
	\centering
	\includegraphics[width=0.67\textwidth]{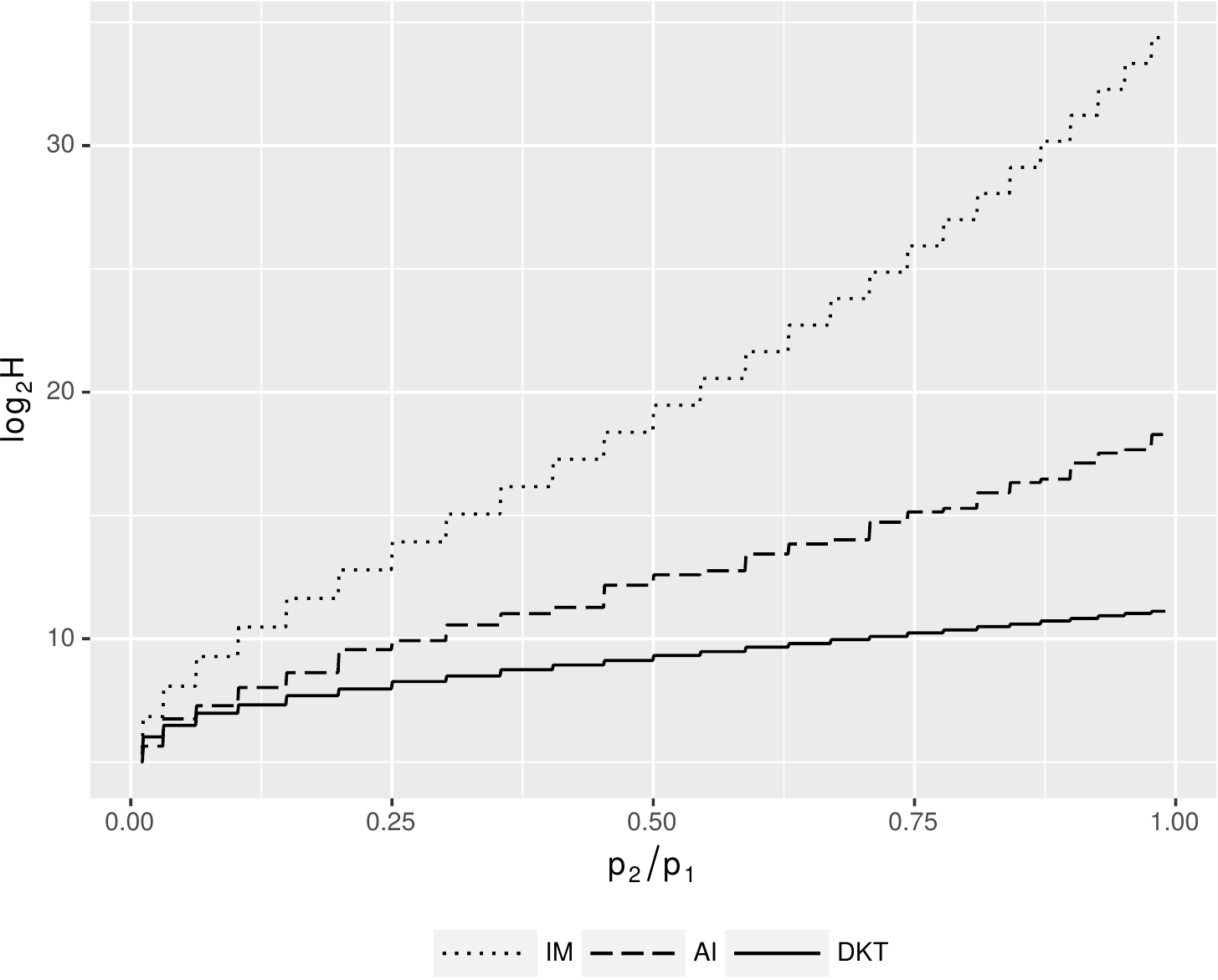}
	\caption{The exact number of locality-sensitive hash functions from a $(r_1, r_2, 0.5, p_2)$-sensitive family used by different frameworks to solve the $(r_1, r_2)$-near neighbor problem on a collection of $2^{30}$ points according to the analysis in this paper.}
	\label{fig:comparison}
\end{figure}
Figure \ref{fig:comparison} reveals that the number of hash functions used by the Indyk-Motwani framework exceeds $2^{30}$, the size of the collection of points $P$, as $p_2$ approaches $p_1$. 
In addition, locality-sensitive hash functions used in practice such as Charikar's SimHash~\cite{charikar2002} and $p$-stable LSH~\cite{datar2004} have evaluation time $O(d)$ for points in $\mathbb{R}^d$. 
These two factors might help explain why a linear scan over sketches of the entire collection of points is a popular approach to solve the approximate near neighbor problem in practice~\cite{weiss2008spectral, gong2012angular}.
The Andoni-Indyk framework reduces the number of hash functions by several orders of magnitude, and the Dahlgaard-Knudsen-Thorup framework presents another improvement of several orders of magnitude. 
Since the word-RAM complexity of the DKT framework matches the the number of hash functions used by the IM framework, the gap between the solid line (DKT) and the dotted line (IM) gives some indication of the time we can spend on evaluating a single hash function in the DKT framework without suffering a noticeable increase in the query time. 
\subsection{Contribution}
\subparagraph*{Improved word-RAM complexity.}
In addition to our work on the Andoni-Indyk and Dahlgaard-Knudsen-Thorup frameworks as mentioned above, we show how the word-RAM complexity of the DKT framework can be reduced by a logarithmic factor. 
The solution is a simple combination of the DKT sampling technique and the AI tensoring technique: 
First we use the DKT sampling technique twice to construct two collections of $\sqrt{L}$ partitions.
Then we use the AI tensoring technique to form $L = \sqrt{L} \times \sqrt{L}$ pairs of partitions from the two collections. 
Below we state our main Theorem \ref{thm:lsh_dkt_ram} in its general form where we make no implicit assumptions about $\LSH$ ($p_1$ and $p_2$ are not assumed to be constant and can depend on for example $n$) or about the complexity of storing a point or a hash function, or computing the distance between pairs of points in the space $(X, \dist)$.
\begin{theorem}\label{thm:lsh_dkt_ram}
	Let $\LSH$ be $(r_1, r_2, p_1, p_2)$-sensitive and let $\rho = \log(1/p_1) / \log(1/p_2)$, then there exists a solution to the $(r_1, r_2)$-near neighbor with the following properties:
\begin{itemize}
	\item The query complexity is dominated by $O(\log_{1/p_2}^{2}(n)/p_1)$ evaluations of functions from $\LSH$, $O(n^\rho)$ distance computations, and $O(n^{\rho} / p_1)$ other word-RAM operations.
	\item The solution uses $O(n^{1 + \rho} /p_1)$ words of space in addition to the space required to store the data and $O(\log_{1/p_2}^{2}(n)/p_1)$ functions from $\LSH$.
\end{itemize}
\end{theorem}

Under the same simplifying assumptions used in the statements of Theorem \ref{thm:lsh_im_simple}, \ref{thm:lsh_ai_simple}, and \ref{thm:lsh_dkt_simple}, our main Theorem \ref{thm:lsh_dkt_ram} can be stated as Theorem \ref{thm:lsh_dkt_simple} with the word-RAM complexity reduced by a logarithmic factor to $O(n^\rho)$. 
This improvement in the word-RAM complexity comes at the cost of a (rather small) constant factor increase in the number of hash functions, lookups, and distance computations compared to the DKT framework.
By varying the size $m$ of the collection of hash functions from $\LSH$ and performing independent repetitions we can obtain a tradeoff between the number of hash functions and the number of lookups.
In Section~\ref{sec:corner} we remark on some possible improvements in the case where $p_2$ is large.

\subparagraph*{Distance sketching using LSH.}
Finally, we combine Theorem \ref{thm:lsh_dkt_ram} with the 1-bit sketching scheme of Li and K{\"o}nig~\cite{li2011theory} where we use the locality-sensitive hash family to create sketches that allow us to leverage word-level parallelism and avoid direct distance computations. 
This sketching technique is well known and has been used before in combination with LSH-based approximate similarity search~\cite{christiani2017scalable}, but we believe there is some value in the simplicity of the analysis and in a clear statement of the combination of the two results as given in Theorem~\ref{thm:lsh_dkt_ram_sketch}, for example in the important case where $0 < p_2 < p_1 < 1$ are constant.
\begin{theorem}\label{thm:lsh_dkt_ram_sketch}
	Let $\LSH$ be $(r_1, r_2, p_1, p_2)$-sensitive and let $\rho = \log(1/p_1) / \log(1/p_2)$, then there exists a solution to the $(r_1, r_2)$-near neighbor with the following properties:
\begin{itemize}
	\item The complexity of the query operation is dominated by $O(\log^2(n)/(p_1 - p_2)^2)$ evaluations of hash functions from $\LSH$ and $O(n^{\rho}/(p_1 - p_2)^2)$ other word-RAM operations.
	\item The solution uses $O(n^{1 + \rho}/p_1 + n/(p_1 - p_2)^2)$ words of space in addition to the space required to store the data and $O(\log^2(n)/(p_1 - p_2)^2)$ hash functions from $\LSH$.
\end{itemize}
\end{theorem}
\section{Preliminaries}\label{sec:preliminaries}
\subparagraph*{Problem and dynamization.}
We begin by defining the version of the approximate near neighbor problem that the frameworks presented in this paper will be solving:
\begin{definition}\label{def:ann}
Let $P \subseteq X$ be a collection of $|P| = n$ points in a distance space $(X, \dist)$.
A solution to the $(r_1, r_2)$-near neighbor problem is a data structure that supports the following query operation:
Given a query point $q \in X$, if there exists a point $x \in P$ with $\dist(q, x) \leq r_1$, 
then, with probability at least $1/2$, return a point $x' \in P$ such that $\dist(q, x') < r_2$.  
\end{definition}
We aim for solutions with a failure probability that is upper bounded by $1/2$.
The standard trick of using $\eta$ independent repetitions of the data structure allows us to reduce the probability of failure to $1/2^\eta$.
For the sake of simplicity we restrict our attention to static solutions, meaning that we do not concern ourselves with the complexity of updates to the underlying set $P$, although it is simple to modify the static solutions presented in this paper to dynamic solutions where the update complexity essentially matches the query complexity~\cite{overmars1981, har-peled2012} 

\subparagraph*{LSH powering.}
The Indyk-Motwani framework and the Andoni-Indyk framework will make use of the following standard powering technique described in the introduction as ``overlaying partitions''.
Let $k \geq 1$ be an integer and let $\LSH$ denote a locality-sensitive family of hash functions as in Definition \ref{def:lsh}. 
We will use the notation $\LSH^k$ to denote the distribution over functions $h' \colon X \to R^k$ where
\begin{equation}
	h'(x) = (h_1(x), \dots, h_k(x))
\end{equation}
and $h_1, \dots, h_k$ are sampled independently at random from $\LSH$.
It is easy to see that $\LSH^k$ is $(r_1, r_2, p_1^k, p_2^k)$-sensitive.
To deal with some special cases we define $\LSH^0$ to be the family consisting of a single constant function.

\subparagraph*{Model of computation.}
We will work in the standard word-RAM model of computation~\cite{hagerup1998} with a word length of $\Theta(\log n)$ bits where $n$ denotes the size of the collection $P$ to be searched in the $(r_1, r_2)$-near neighbor problem.
During the preprocessing stage of our solutions we will assume access to a source of randomness that allows us to sample independently from a family $\LSH$ and to seed pairwise independent hash functions~\cite{carter1979, wegman1981}. 
The latter can easily be accomplished by augmenting the model with an instruction that generates a uniformly random word in constant time and using that to seed the tables of a Zobrist hash function~\cite{zobrist1970new}.
\section{Frameworks}\label{sec:frameworks}
\subparagraph*{Overview.}
We will describe frameworks that take as input a $(r_1, r_2, p_1, p_2)$-sensitive family $\LSH$ and a collection $P$ of $n$ points and constructs a data structure that solves the $(r_1, r_2)$-near neighbor problem.
The frameworks described in this paper all use the same high-level technique of constructing $L$ hash functions $g_{1},\dots,g_{L}$ that are used to partition space such that a pair of points $x, y$ with $\dist(x, y) \leq r_1$ will end up in the same part of one of the $L$ partitions with probability at least $1/2$. 
That is, for $x, y$ with $\dist(x, y) \leq r_1$ we have that $\Pr[\exists l \in [L] \colon g_{l}(x) = g_{l}(y)] \geq 1/2$ where $[L]$ is used to denote the set $\{1,2,\dots,L\}$.
At the same time we ensure that the expected number of collisions between pairs of points $x,y$ with $\dist(x, y) \geq r_2$ is at most one in each partition.

\subparagraph*{Preprocessing and queries.}
During the preprocessing phase, for each of the $L$ hash functions $g_{1}, \dots,g_{L}$ we compute the partition of the collection of points $P$ induced by $g_{l}$ and store it in a hash table in the form of key-value pairs $(z, \{ x \in P \mid g_{l}(x) = z \})$.
To reduce space usage we store only a single copy of the collection $P$ and store references to $P$ in our $L$ hash tables.
To guarantee lookups in constant time we can use the perfect hashing scheme by Fredman et al.~\cite{fredman1984storing} to construct our hash tables.
We will assume that hash values $z = g_{l}(x)$ fit into $O(1)$ words. 
If this is not the case we can use universal hashing~\cite{carter1977} to operate on fingerprints of the hash values.

We perform a query for a point $q$ as follows: for $l = 1, \dots, L$ we compute $g_{l}(q)$, 
retrieve the set of points $\{ x \in P \mid g_{l}(x) = g_{l}(q) \}$, and compute the distance between $q$ and each point in the set.
If we encounter a point $x'$ with $\dist(q, x') < r_2$ then we return $x'$ and terminate.
If after querying the $L$ sets no such point is encountered we return a special symbol $\varnothing$ and terminate.

We will proceed by describing and analyzing the solutions to the $(r_1, r_2)$-near neighbor problem for different approaches to sampling, storing, and computing the $L$ hash functions $g_{1}, \dots, g_{L}$, resulting in the different frameworks as mentioned in the introduction.
\subsection{Indyk-Motwani}\label{sec:im}
To solve the $(r_1, r_2)$-near neighbor problem using the Indyk-Motwani framework we sample $L$ hash functions $g_{1}, \dots,g_{L}$ independently at random from the family $\LSH^k$ where we set $k = \lceil \log(n) / \log(1/p_2) \rceil$ and $L = \lceil (\ln 2)/p_1^k \rceil$.
Correctness of the data structure follows from the observation that the probability that a pair of points $x, y$ with $\dist(x,y) \leq r_1$ does not collide under a randomly sampled $g_l \sim \LSH^k$ is at most $1 - p_1^k$. 
We can therefore upper bound the probability that a near pair of points does not collide under any of the hash functions by $(1-p_1^k)^L \leq \exp(-p_1^k L) \leq 1/2$ using a standard bound stated as Lemma~\ref{lem:exp_upper} in Appendix~\ref{app:inequalities}.

In the worst case, the query operation computes $L$ hash functions from $\LSH^k$ corresponding to $Lk$ hash functions from $\LSH$. 
For a query point $q$ the expected number of points $x' \in P$ with $\dist(q, x') \geq r_2$ that collide with $q$ under a randomly sampled $g_l \sim \LSH^k$ is at most $np_2^k \leq np_2^{\log(n) / \log (1/p_2)} = 1$.
It follows from linearity of expectation that the total expected number of distance computations during a query is at most $L$.
The result is summarized in Theorem \ref{thm:lsh_im_exact} from which the simplified Theorem \ref{thm:lsh_im_simple} follows.
\begin{theorem}[Indyk-Motwani {\cite{indyk1998, har-peled2012}}]\label{thm:lsh_im_exact}
Given a $(r_1, r_2, p_1, p_2)$-sensitive family $\LSH$ we can construct a data structure that solves the $(r_1, r_2)$-near neighbor problem such that for
$k = \lceil \log(n) / \log(1/p_2) \rceil$ and $L = \lceil (\ln 2)/p_1^k \rceil$ the data structure has the following properties:
\begin{itemize}
\item The query operation uses at most $Lk$ evaluations of hash functions from $\LSH$, 
	expected $L$ distance computations, and $O(Lk)$ other word-RAM operations.
\item The data structure uses $O(nL)$ words of space in addition to the space required to store the data and $Lk$ hash functions from $\LSH$.
\end{itemize}
\end{theorem}
Theorem \ref{thm:lsh_im_exact} gives a bound on the expected number of distance computations while the simplified version stated in Theorem \ref{thm:lsh_im_simple} uses Markov's inequality and independent repetitions to remove the expectation from the bound by treating an excessive number of distance computations as a failure.
\subsection{Andoni-Indyk}\label{sec:ai}
In 2006 Andoni and Indyk, as part of a paper on the substring near neighbor problem, introduced an improvement to the Indyk-Motwani framework that reduces the number of locality-sensitive hash functions~\cite{andoni2006efficient}.
Their improvement comes from the use of a technique that we will refer to as tensoring: setting the hash functions $g_1, \dots, g_L$ to be all $t$-tuples from a collection of $m$ functions sampled from $\LSH^{k/t}$ where $m \ll L$.
The analysis in~\cite{andoni2006efficient} shows that by setting $m = n^{\rho/t}$ and repeating the entire scheme $t!$ times, the total number of hash functions can be reduced to $O(\exp(\sqrt{\rho \log n \log \log n}))$ when setting $t = \sqrt{\frac{\rho \log n}{\log \log n}}$.
This analysis ignores integer constraints on $t$, $k$, and $m$, and implicitly place restrictions on $p_1$ and $p_2$ in relation to $n$ (e.g.\ $0 < p_2 < p_1 < 1$ are constant).
We will introduce a slightly different scheme that takes into account integer constraints and analyze it without restrictions on the properties of $\LSH$.

Assume that we are given a $(r_1, r_2, p_1, p_2)$-sensitive family $\LSH$.
Let $\eta, t, k_1, k_2, m_1, m_2$ be non-negative integer parameters.
Each of the $L$ hash functions $g_1, \dots, g_L$ will be formed by concatenating one hash function from each of $t$ collections of $m_1$ hash functions from $\LSH^{k_1}$ and concatenating a last hash function from a collection of $m_2$ hash functions from $\LSH^{k_2}$.
We take all $m_1^t m_2$ hash functions of the above form and repeat $\eta$ times for a total of $L = \eta m_1^t m_2$ hash functions constructed from a total of $H = \eta(m_1 k_1 t + m_2 k_2)$ hash functions from $\LSH$.
In Appendix \ref{app:ai} we set parameters, leaving $t$ variable, and provide an analysis of this scheme, showing that $L$ matches the Indyk-Motwani framework bound of $O(1/p_1^k)$ up to a constant where $k = \lceil \log(n) / \log(1/p_2) \rceil$ as in Theorem \ref{thm:lsh_im_exact}.

\subparagraph*{Setting $t$.}
It remains to show how to set $t$ to obtain a good bound on the number of hash functions $H$.
Note that in practice we can simply set $t = \argmin_t H$ by trying $t = 1,\dots,k$. 
If we ignore integer constraints and place certain restrictions of $\LSH$ as in the original tensoring scheme by Andoni and Indyk we want to set $t$ to minimize the expression $t^t n^{\rho/t}$. 
This minimum is obtained when setting $t$ such that $t^2 \log t = \rho \log n$.
We therefore cannot do much better than setting $t = \sqrt{\rho \log(n) / \log \log n}$ which gives the bound $H = O(\exp(\sqrt{\rho \log(n) \log \log n}))$ as shown in \cite{andoni2006efficient}. 
To allow for easy comparison with the Indyk-Motwani framework without placing restrictions on $\LSH$ we set $t = \lceil \sqrt{k} \rceil$, resulting in Theorem \ref{thm:tensoring}. 
\begin{theorem}\label{thm:tensoring}
Given a $(r_1, r_2, p_1, p_2)$-sensitive family $\LSH$ we can construct a data structure that solves the $(r_1, r_2)$-near neighbor problem such that for $k = \lceil \log(n) / \log(1/p_2) \rceil$, $H = k(\sqrt{k}/p_1)^{\sqrt{k}}$, and $L = \lceil 1/p_1^k \rceil$ the data structure has the following properties:
\begin{itemize}
	\item The query operation uses $O(H)$ evaluations of functions from $\LSH$, $O(L)$ distance computations, and $O(L + H)$ other word-RAM operations.
	\item The data structure uses $O(nL)$ words of space in addition to the space required to store the data and $O(H)$ hash functions from $\LSH$.
\end{itemize}
\end{theorem}
Thus, compared to the Indyk-Motwani framework we have gone from using $O(k(1/p_1)^k)$ locality-sensitive hash functions to $O(k(\sqrt{k}/p_1)^{\sqrt{k}})$ locality-sensitive hash functions. 
Figure~\ref{fig:comparison} shows the actual number of hash functions of the revised version of the Andoni-Indyk scheme as analyzed in Appendix \ref{app:ai} when $t$ is set to minimize $H$.
\subsection{Dahlgaard-Knudsen-Thorup}\label{sec:dkt}
In a recent paper Dahlgaard et al.~\cite{dahlgaard2017fast} introduce a different technique for reducing the number of locality-sensitive hash functions.
The idea is to construct each hash value $g_{l}(x)$ by sampling and concatenating $k$ hash values from a collection of $km$ pre-computed hash functions from~$\LSH$.
Dahlgaard et al.\ applied this technique to provide a fast solution the the approximate near neighbor problem for sets under Jaccard similarity.
In this paper we use the same technique to derive a general framework solution that works with every family of locality-sensitive hash functions, reducing the number of locality-sensitive hash functions compard to the Indyk-Motwani and Andoni-Indyk frameworks.

Let $[n]$ denote the set of integers $\{1,2,\dots,n\}$.
For $i \in [k]$ and $j \in [m]$ let $h_{i,j} \sim \LSH$ denote a hash function in our collection.
To sample from the collection we use $k$ pairwise independent hash functions~\cite{wegman1981} of the form $f_i \colon [L] \to [m]$ and set
\begin{equation*}
g_{l}(x) = (h_{1,f_{1}(l)}(x), \dots, h_{k,f_{k}(l)}(x)).
\end{equation*}
To show correctness of this scheme we will use make use of an elementary one-sided version of Chebyshev's inequality stating that for a random variable $Z$ with mean $\mu > 0$ and variance $\sigma^2 < \infty$ we have that $\Pr[Z \leq 0] \leq \sigma^2 / (\mu^2 + \sigma^2)$. 
For completeness we have included the proof of this inequality in Lemma~\ref{lem:cantelli} in Appendix~\ref{app:inequalities}.
We will apply this inequality to lower bound the probability that there are no collisions between close pairs of points.
For two points $x$ and $y$ let $Z_l = \1\{g_{l}(x) = g_{l}(y)\}$ so that $Z = \sum_{l = 1}^L Z_l$ denotes the sum of collisions under the $L$ hash functions. 
To apply the inequality we need to derive an expression for the expectation and the variance of the random variable $Z$.
Let $p = \Pr_{h \sim \LSH}[h(x) = h(y)]$ then by linearity of expectation we have that $\mu = \E[Z] = L p^k$.
To bound $\sigma^2 = \E[Z^2] - \mu^2$ we proceed by bounding $\E[Z^2]$ where we note that $Z_l = \Pi_{i = 1}^{k} Y_{l, i}$ for $Y_{l, i} = 1\{h_{i,f_{i}(l)}(x) = h_{i,f_{i}(l)}(x)\}$ and make use of the independence between $Y_{l,i}$ and $Y_{l', i'}$ for $i \neq i'$.
\begin{align*}
	\E[Z^2] &= \sum_{\substack{l, l' \in [L] \\ l \neq l'}} \E[Z_l Z_{l'}] + \sum_{l = 1}^L \E[Z_l] \\ 
			&= (L^2 - L) \E[Z_l Z_{l'}] + \mu \\
			&\leq L^2 \E\left[ \Pi_{i = 1}^{k} Y_{l, i} Y_{l', i} \right] + \mu \\ 
			&= L^2 \left( \E[Y_{l, i} Y_{l', i}] \right)^k + \mu. 
\end{align*}
We have that $\E[Y_{l, i} Y_{l', i}] = \Pr[f_i(l) = f_i(l')] p + \Pr[f_i(l) \neq f_i(l')] p^2 = (1/m)p + (1-1/m)p^2$ which follows from the pairwise independence of $f_i$. 
Let $\varepsilon > 0$ and set $m = \lceil \frac{1-p_1}{p_1} \frac{k}{\ln(1+\varepsilon)} \rceil$ then for $p \geq p_1$ we have that $\left( \E[Y_{l, i} Y_{l', i}] \right)^k \leq (1 + \varepsilon)p^{2k}$.
This allows us to bound the variance of $Z$ by $\sigma^2 \leq \varepsilon \mu^2 + \mu$ resulting in the following lower bound on the probability of collision between similar points.
\begin{lemma} \label{lem:dkt_success}
For $\varepsilon > 0$ let $m \geq \lceil \frac{1-p_1}{p_1} \frac{k}{\ln(1+\varepsilon)} \rceil$, then for every pair of points $x, y$ with $\dist(x, y) \leq r_1$ we have that 
\begin{equation}
	\Pr[\exists l \in [L] \colon g_{l}(x) = g_{l}(y)] \geq \frac{1 + \varepsilon \mu}{1 + (1 + \varepsilon)\mu}.
\end{equation}
\end{lemma}
By setting $\varepsilon = 1/4$ and $L = \lceil (2 \ln(2))/p_1^k \rceil$ we obtain an upper bound on the failure probability of $1/2$.
Setting the size of each of the $k$ collections of pre-computed hash values to $m = \lceil 5k/p_1 \rceil$ is sufficient to yield the following solution to the $(r_1, r_2)$-near neighbor problem where provide exact bounds on the number of lookups $L$ and hash functions $H$:
\begin{theorem}[Dahlgaard-Knudsen-Thorup {\cite{dahlgaard2017fast}}]\label{thm:lsh_dkt_exact}
Given a $(r_1, r_2, p_1, p_2)$-sensitive family $\LSH$ we can construct a data structure that solves the $(r_1, r_2)$-near neighbor problem such that for $k = \lceil \log(n) / \log(1/p_2) \rceil$, $H = k \lceil 5k / p_1 \rceil$, and $L = \lceil (2 \ln(2))/p_1^k \rceil$ the data structure has the following properties:
\begin{itemize}
\item The query operation uses at most $H$ evaluations of hash functions from $\LSH$, 
	expected $L$ distance computations, and $O(Lk)$ other word-RAM operations.
\item The data structure uses $O(nL)$ words of space in addition to the space required to store the data and $H$ hash functions from $\LSH$.
\end{itemize}
\end{theorem}
Compared to the Indyk-Motwani framework we have reduced the number of locality-sensitive hash functions $H$ from $O(k (1/p_1)^k)$ to $O(k^2 / p_1)$ at the cost of using twice as many lookups. 
To reduce the number of lookups further we can decrease $\varepsilon$ and perform several independent repetitions.
This comes at the cost of an increase in the number of hash functions $H$.
\section{Reducing the word-RAM complexity} \label{section:word-RAM}
One drawback of the DKT framework is that each hash value $g_{l}(x)$ still takes $O(k)$ word-RAM operations to compute, even after the underlying locality-sensitive hash functions are known. 
This results in a bound on the total number of additional word-RAM operations of $O(Lk)$.
We show how to combine the DKT universal hashing technique with the AI tensoring technique to ensure that the running time is dominated by $O(L)$ distance computations and $O(H)$ hash function evaluations. 
The idea is to use the DKT scheme to construct two collections of respectively $L_1$ and $L_2$ hash functions, and then to use the AI tensoring approach to form $g_1, \dots, g_L$ as the $L = L_1 \times L_2$ combinations of functions from the two collections.
The number of lookups can be reduced by applying tensoring several times in independent repetitions, but for the sake of simplicity we use a single repetition.
For the usual setting of $k = \lceil \log(n) / \log(1/p_2) \rceil$ let $k_1 =  \lceil k/2 \rceil$ and $k_2 = \lfloor k/2 \rfloor$.
Set $L_1 = \lceil 6 (1/p_1)^{k_1} \rceil$ and $L_2 = \lceil 6 (1/p_1)^{k_2} \rceil$.
According to Lemma \ref{lem:dkt_success} if we set $\varepsilon = 1/6$ the success probability of each collection is at least $3/4$ and by a union bound the probability that either collection fails to contain a colliding hash function is at most $1/2$.
This concludes the proof of our main Theorem \ref{thm:lsh_dkt_ram}.
\subsection{Sketching}\label{sec:sketching}
The theorems of the previous section made no assumptions on the word-RAM complexity of distance computations and instead stated the number of distance computations as part of the query complexity. 
We can use a $(r_1, r_2, p_1, p_2)$-sensitive family $\LSH$ to create sketches that allows us to efficiently approximate the distance between pairs of points, provided that the gap between $p_1$ and $p_2$ is sufficiently large. 
In this section we will re-state the results of Theorem~\ref{thm:lsh_dkt_ram} when applying the family $\LSH$ to create sketches using the 1-bit sketching scheme of Li and König~\cite{li2011theory}.
Let $b$ be a positive integer denoting the length of the sketches in bits.
The advantage of this scheme is that we can use word level parallelism to evaluate a sketch of $b$ bits in time $O(b/\log n)$ in our word-RAM model with word length $\Theta(\log n)$. 

For $i = 1, \dots, b$ let $h_i \colon X \to R$ denote a randomly sampled locality-sensitive hash function from $\LSH$ and let $f_i \colon R \to \{0,1\}$ denote a randomly sampled universal hash function.
We let $s(x) \in \{0,1\}^b$ denote the sketch of a point $x \in X$ where we set the $i$th bit of the sketch $s(x)_i = f_i(h(x))$.
For two points $x, y \in X$ the probability that they agree on the $i$th bit is $1$ if the points collide under $h_i$ and $1/2$ otherwise. 
\begin{equation*}
	\Pr[s(x)_i = s(y)_i] = \Pr[h_i(x) = h_i(y)] + (1 - \Pr[h_i(x) = h_i(y)])/2 = (1 + \Pr[h_i(x) = h_i(y)])/2.
\end{equation*}
We will apply these sketches during our query procedure instead of direct distance computations when searching through the points in the $L$ buckets, comparing them to our query point $q$.
Let $\lambda \in (0,1)$ be a parameter that will determine whether we report a point or not.
For sketches of length $b$ we will return a point $x$ if $\norm{s(q) - s(x)}_1 > \lambda b$.
An application of Hoeffiding's inequality gives us the following properties of the sketch:
\begin{lemma}\label{lem:sketching}
	Let $\LSH$ be a $(r_1, r_2, p_1, p_2)$-sensitive family and let $\lambda = (1+p_2)/2 + (p_1-p_2)/4$, 
	then for sketches of length $b \geq 1$ and for every pair points $x, y \in X$:  
	\begin{itemize}
		\item If $\dist(x ,y) \leq r_1$ then $\Pr[\norm{s(x) - s(y)}_1 \leq \lambda b] \leq e^{b(p_1 - p_2)^2 / 8 }$.
		\item If $\dist(x ,y) \geq r_2$ then $\Pr[\norm{s(x) - s(y)}_1 > \lambda b] \leq e^{b(p_1 - p_2)^2 / 8 }$.
	\end{itemize}
\end{lemma}

If we replace the exact distance computations with sketches we want to avoid two events: 
Failing to report a point with $\dist(q,x) \leq r_1$ and reporting a point $x$ with $\dist(q, x) \geq r_2$.
By setting $b = O(\ln(n) / (p_1 - p_2)^2)$ and applying a union bound over the $n$ events that the sketch fails for a point in our collection $P$ we obtain Theorem~\ref{thm:lsh_dkt_ram_sketch}. 
\section{The number of hash functions in corner cases}\label{sec:corner}
When the collision probabilities of the $(r_1, r_2, p_1, p_2)$-sensitive family $\LSH$ are close to one we get the behavior displayed in Figure \ref{fig:p1_090} where we have set $p_1 = 0.9$. 
Here it may be possible to reduce the number of hash functions by applying the DKT framework to the family $\LSH^\tau$ for some positive integer $\tau$. That is, instead of applying the DKT technique directly to $\LSH$ we first apply the powering trick to produce the family $\LSH^\tau$.
The number of locality-sensitive hash functions from $\LSH$ used by the DKT framework is given by $H = O( (\log(n) / \log(1/p_2))^2 / p_1)$.
If we instead use the family $\LSH^\tau$ the expression becomes $H = O( \tau(\log(n) / \log(1/p_2^\tau))^2 / p_1^\tau) = O((\log(n) / \log(1/p_2))^2 / \tau p_1^\tau)$. 
Ignoring integer constraints, the value of $\tau$ that maximizes $\tau p_1^\tau$, thereby minimizing $H$, is given by $\tau = 1 / \ln(1/p_1)$.
Discretizing, the resulting number of hash functions when setting $\tau = \lceil 1 / \ln(1/p_1)\rceil$ is given by $H = O(\rho (\log n)^2 / (p_1 \log(1/p_2)))$.
For constant $\rho$ and large $p_2$ this reduces the number of hash functions by a factor $1/\log(1/p_2)$.
\begin{figure}
	\centering
	\includegraphics[width=0.67\textwidth]{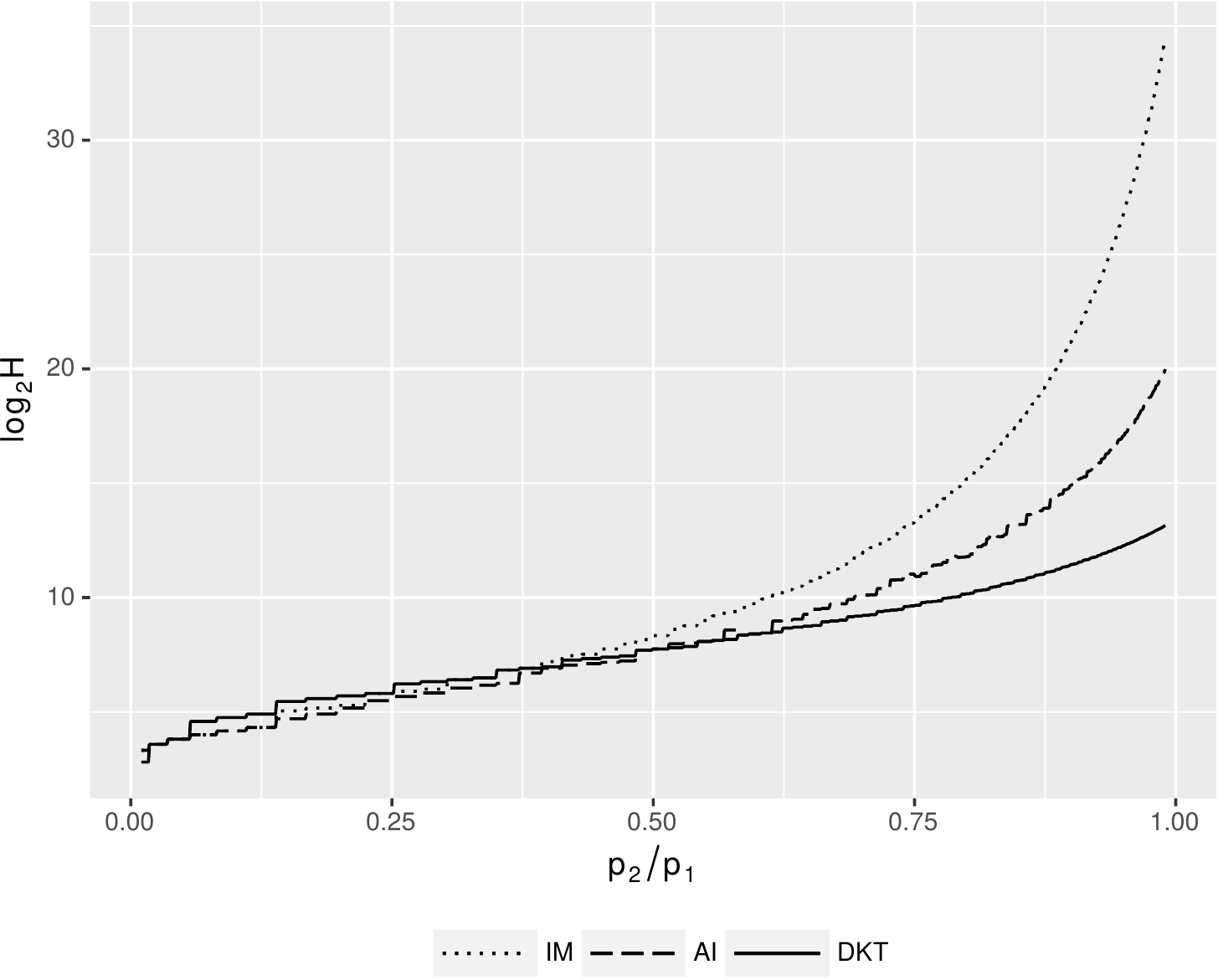}
	\caption{The number of locality-sensitive hash functions from a $(r_1, r_2, 0.9, p_2)$-sensitive family used by different frameworks to solve the $(r_1, r_2)$-near neighbor problem on a collection of $2^{30}$ points.}
	\label{fig:p1_090}
\end{figure}
The behavior for small values of $p_1$ is displayed in Figure \ref{fig:p1_010} where we have set $p_1 = 0.1$.
\begin{figure}
	\centering
	\includegraphics[width=0.67\textwidth]{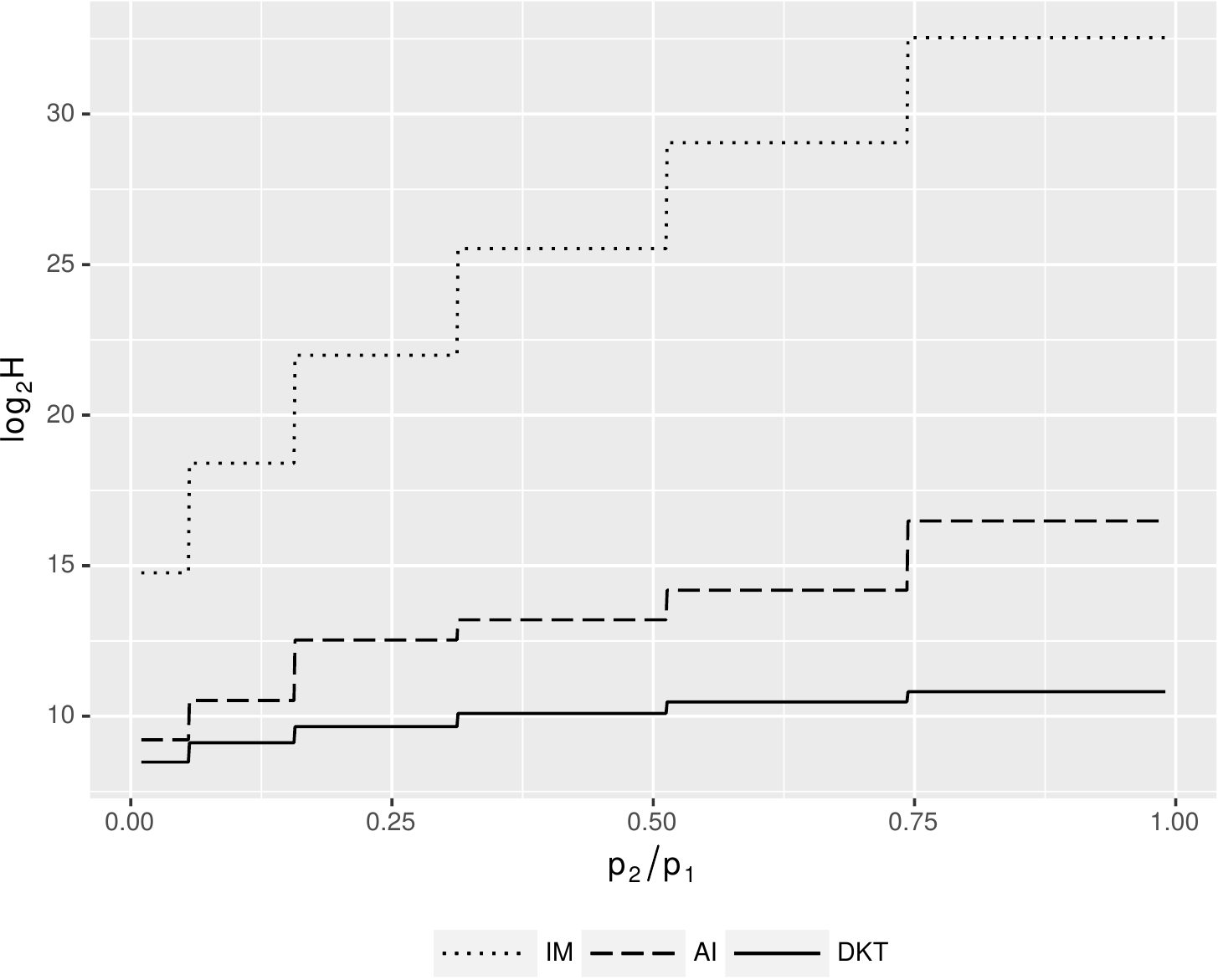}
\caption{The number of locality-sensitive hash functions from a $(r_1, r_2, 0.1, p_2)$-sensitive family used by different frameworks to solve the $(r_1, r_2)$-near neighbor problem on a collection of $2^{30}$ points.}
\label{fig:p1_010}
\end{figure}
\section{Conclusion and open problems}
We have shown that there exists a simple and general framework for solving the $(r_1, r_2)$-near neighbor problem using only few locality-sensitive hash functions and with a reduced word-RAM complexity matching the number of lookups.
The analysis in this paper indicates that the performance of the Dahlgaard-Knudsen-Thorup framework is highly competitive compared to the Indyk-Motwani framework in practice, especially when locality-sensitive hash functions are expensive to evaluate, as is often the case.

An obvious open problem is the question of whether the number of locality-sensitive hash functions can be reduced even below $O(k^2 / p_1)$.
Another possible direction for future research would be to obtain similar framework results in the context of solutions to the $(r_1, r_2)$-near neighbor problem that allow for space-time tradeoffs~\cite{andoni2017optimal, christiani2017framework}.
\subparagraph*{Acknowledgements.}
I want to thank Rasmus Pagh commenting on an earlier version of this manuscript and for making me aware of the application of the tensoring technique in~\cite{sundaram2013streaming} that led me to the Andoni-Indyk framework~\cite{andoni2006efficient}.
\appendix
\section{Inequalities}\label{app:inequalities}
We make use of the following standard inequalities for the exponential function. 
See \cite[Chapter 3.6.2]{mitrinovic1970} for more details.
\begin{lemma}\label{lem:exp_upper}
	Let $n, t \in \mathbb{R}$ such that $n \geq 1$ and $|t| \leq n$ then $e^{-t}(1-t^2 /n) \leq (1 - t/n)^n \leq e^{-t}$.
\end{lemma}
\begin{lemma}\label{lem:exp_taylor}
	For $t \geq 0$ we have that $e^{-t} \leq 1 - t + t^2 / 2$.
\end{lemma}

We make use of a one-sided version of Chebyshev's inequality to show correctness of the Dahlgaard-Knudsen-Thorup LSH framework. 
\begin{lemma}[Cantelli's inequality]\label{lem:cantelli}
Let $Z$ be a random variable with $\E[Z] = \mu > 0$ and $\Var[Z] = \sigma^2 < \infty$ then $\Pr[Z \leq 0] \leq \sigma^2/(\mu^2 + \sigma^2)$. 
\end{lemma}
\begin{proof}
	For every $s \in \mathbb{R}$ we have that 
\begin{equation*}
\Pr[Z \leq 0] = \Pr[-(Z - \mu) + s \geq \mu + s] \leq  \Pr[(-(Z - \mu) + s)^2 \geq (\mu + s)^2]. 
\end{equation*}
Next we apply Markov's inequality 
\begin{equation*}
\Pr[(-(Z - \mu) + s)^2 \geq (\mu + s)^2] \leq \E[(-(Z - \mu) + s)^2]/ (\mu + s)^2 = (\sigma^2 + s^2)/ (\mu + s)^2 
\end{equation*}
Set $s = \sigma^2 / \mu$ and use that $\sigma^2 = s\mu$ to simplify 
\begin{equation*}
(\sigma^2 + s^2)/ (\mu + s)^2 = (s\mu + s^2)/ (\mu + s)^2 = \sigma^2 / (\mu^2 + \sigma^2).
\end{equation*}
\end{proof}

To analyze the 1-bit sketching scheme by Li and K{\"o}nig we make use of Hoeffding's inequality:
\begin{lemma}[{Hoeffding \cite[Theorem 1]{hoeffding1963}}] \label{lem:hoeffding}
	Let $X_1, X_2, \dots, X_n$ be independent random variables satisfying $0 \leq X_i \leq 1$ for $i \in [n]$.
	Define $\bar{X} = (X_1 + X_2 + \dots + X_n)/n$ and $\mu = \E[\bar{X}]$, then:
	\begin{itemize}
		\item[-] For $0 < \varepsilon < 1 - \mu$ we have that $\Pr[\bar{X} - \mu \geq \varepsilon] \leq e^{-2n\varepsilon^{2}}$.
		\item[-] For $0 < \varepsilon < \mu$ we have that $\Pr[\bar{X} - \mu \leq - \varepsilon] \leq e^{-2n\varepsilon^{2}}$.
	\end{itemize}
\end{lemma}
\section{Analysis of the Andoni-Indyk framework}\label{app:ai}
Let $\varphi$ denote the probability that a pair of points $x, y$ with $\dist(x,y) \leq r_1$ collide in a single repetition of the scheme.
A collision occurs if and only if there there exists at least one hash function in each of the underlying $t + 1$ collections where the points collide.
It follows that
\begin{equation*}
	\varphi = (1 - (1-p_1^{k_1})^{m_1})^t (1 - (1-p_1^{k_2})^m_2).
\end{equation*}
To guarantee a collision with probability at least $1/2$ it suffices to set $\eta = \lceil \ln(2) / \varphi \rceil$.

We will proceed by analyzing this scheme where we let $t \geq 1$ be variable and set parameters as followers:
\begin{align*}
	k &= \lceil \log(n) / \log(1/p_2) \rceil \\
	k_1 &= \lfloor k/t \rfloor \\
	k_2 &= k - tk_1 \\
	m_1 &= \lceil 1/ t p_1^{k_1} \rceil \\
	m_2 &= \lceil 1/ p_1^{k_2} \rceil \\
    \eta &= \lceil \ln(2) / \varphi \rceil.
\end{align*}
To upper bound $L$ we begin by lower bounding $\varphi$.
The second part of $\varphi$ can be lower bounded using Lemma \ref{lem:exp_upper} to yield $(1 - (1-p_1^{k_2})^{m_2}) \geq 1 - 1/e$.
To lower bound $(1 - (1-p_1^{k_1})^{m_1})^t$ we first note that in the case where $p_1^{k_1} > 1/t$ we have $m_1 = 1$ and the expression can be lower bounded by $p_1^{k_1 t} = (p_1^{k_1} m_1)^t \geq  (p_1^{k_1} m_1)^t / 2e$.
The same lower bound holds in the case there $t = 1$.
In the case where $p_1^{k_1} \leq 1/t$ and $t \geq 2$ we make use of Lemma \ref{lem:exp_upper} and \ref{lem:exp_taylor} to derive the lower bound. 
\begin{align*}
	1 - (1-p_1^{k_1})^{m_1} &\geq 1 - e^{-p_1^{k_1 m_1}} \\
 &\geq 1 - (1 - p_1^{k_1} m_1 + (p_1^{k_1} m_1)^2/2) \\
 &\geq p_1^{k_1} m_1 (1 - p_1^{k_1}(1/t p_1^{k_1} + 1)/2) \\
 &\geq p_1^{k_1} m_1 (1 - 1/t). 
\end{align*}
Using the bound $(p_1^{k_1} m_1 (1 - 1/t))^t \geq (p_1^{k_1} m_1)^t / 2e$ we have that 
\begin{equation*}
	\varphi \geq (p_1^{k_1} m_1)^t / 4e \geq (1/t)^t / 4e.
\end{equation*}
We can then bound the number of lookups and the expected number of distance computations
\begin{equation*}
	L = \eta m_1^t m_2 \leq (4e / (p_1^{k_1} m_1)^t + 1) m_1^t (1/p_1^{k_2} + 1) \leq 16e (1/p_1^k). 
\end{equation*}
Note that this matches the upper bound of the Indyk-Motwani LSH framework up to a constant factor.

To bound the number of hash functions from $\LSH$ we use that $k_1 \leq k/t \leq k$ and $k_2 < t$.
\begin{equation*}
	H = \eta (m_1 k_1 t + m_2 k_2) \leq 8e t^t \left(\frac{k}{t p_1^{k/t}} + \frac{t-1}{p_1^{t-1}}\right).
\end{equation*}

\bibliography{fast}
\end{document}